\documentclass[10pt]{article}
\usepackage{amsmath,amsthm}
\usepackage{amsfonts}
\usepackage{bbm}
\usepackage[active]{srcltx}
\usepackage{hyperref}

\usepackage{amssymb,epsfig}
\newtheorem{thrm}{Theorem}[section]
\newtheorem{defn}[thrm]{Definition}

\newtheorem{asmp}[thrm]{Assumption}

\newtheorem{remark}[thrm]{Remark}

\newtheorem{prop}[thrm]{Proposition}
\usepackage{color}

\newcommand{\cred}{\color{red}}

\newcommand{\re}{\mathbb{R}}
\newcommand{\R}{\mathcal{R}}
\newcommand{\A}{\mathcal{A}}
\newcommand{\p}{\mathbb{P}}

\newcommand{\Q}{\mathbb{Q}}
\newcommand{\e}{\mathbb{E}}

\newcommand{\F}{\mathcal{F}}

\newcommand{\td}{\mathrm{d}}
\newcommand{\cer}{\mathrm{CER}}
\newcommand{\tr}{'}

\DeclareMathOperator{\vax}{\xrightarrow[x\to\infty]{}}
\newcommand{\ovl}{\overline}
\usepackage{color}

\title{Optimal portfolios of a long--term investor with floor or drawdown constraints}

\author{Vladimir Cherny\thanks{Mathematical Institute and the Oxford-Man Institute of Quantitative Finance, University of Oxford, OX1 3LB, UK. E-mail address: \texttt{Vladimir.Cherny@maths.ox.ac.uk}. Research supported by 
the Clarendon scholarship at the University of Oxford.} \and Jan Ob\l\'oj\thanks{Mathematical Institute, St John's College and the Oxford-Man Institute of Quantitative Finance, University of Oxford, OX1 3LB, UK. E-mail address: \texttt{jan.obloj@maths.ox.ac.uk}}\and
\smallskip\\
\multicolumn{1}{p{.7\textwidth}}{\centering\emph{University of Oxford}}}
 
\date{May 2013}

\begin{document}

\maketitle

\begin{abstract}
We study the portfolio selection problem of a long-run investor who is maximising the asymptotic growth rate of her expected utility. We show that, somewhat surprisingly, it is essentially not affected by introduction of a floor constraint which requires the wealth process to dominate a given benchmark at all times. We further study the notion of long-run optimality of wealth processes via convergence of finite horizon value functions to the asymptotic optimal value. We characterise long-run optimality under floor and drawdown constraints.
\end{abstract}

\section{Introduction}

This paper considers a dynamic asset allocation problem of a risk-sensitive investor focusing on problems related to long-run optimality and presence of pathwise constraints on investor's wealth process. More specifically, we consider drawdown constraints and floor constraints and are interested in long-run investor choosing her investment strategy $V$ according to
\begin{equation}\label{eq:cer_def}
\sup_{V }\R_U(V), \quad \text{ where } \R_{U}(V):=\limsup_{T\to \infty} \frac{1}{T}\log \e\left[U\left(V_T\right)\right].
\end{equation}
The idea to look at the maximisation of the growth rate of expected utility $\R_U(V)$ goes back to Dumas and Luciano \cite{DumasLuciano:91}, Grossman and Vila \cite{GrossmanVila:92} and Grossman and Zhou \cite{GZ}.
The criterion is designed to capture the long-horizon optimality and is often more tractable
 than the fixed-horizon utility maximisation of terminal wealth, cf.~Guasoni and Robertson \cite{GuasoniRobertson:11}.
The optimal rate is called the \emph{Certainty Equivalent Rate}\footnote{Also \emph{Equivalent Safe Rate}, see Guasoni and Ob\l\'oj \cite{GuasoniObloj}.} (CER) and has the interpretation of a critical incentive rate -- if the investor was offered such (or higher) rate of growth via other investment opportunities she would be happy to abandon the market and move to the alternative investment opportunities.
The above criterion has also natural links with the \emph{risk-sensitive control}, see e.g.\ Bielecki and Pliska \cite{BieleckiPliska:99} or Fleming and Sheu \cite{FleSheu}. Indeed, considering the power utility $U(x)=x^p/p$, $F_T=\log V_T$ and expanding around $p=0$ we obtain
$$\frac{1}{p}\log \e\left[\mathrm{e}^{pF_T}\right] = \e\left[ F_T\right] + \frac{p}{2} \textrm{Var}(F_T) + O(p^2),$$
where Var denotes the variance and we expanded first $\exp(\cdot)$ and then $\log(1+\cdot)$. Of special interest is the case $p<0$ which captures the tradeoff between maximisation of the growth rate of returns $F$ and controlling their variance.

The first problem considered in this paper corresponds to \eqref{eq:cer_def} with $V$ spanning admissible wealth processes which further satisfy $V_t\geq G_t$, $t\geq 0$, where $G$ is a given floor process. Such a floor constraint is also referred to in the literature as an American capital guarantee, see El Karoui and Meziou \cite{KaMe,KaMeMax}. It is motivated by different insurance products available in a real-world financial markets which  guarantee a pre-specified minimum wealth for the investor's portfolio. In particular, the popular CPPI (Constant Proportion Portfolio Insurance) strategies guarantee that discounted value of the wealth dominates a pre-specified floor, see Black and Jones \cite{BlJo} and Perold and Sharpe \cite{PeSh}.

Our main result states that for any reasonable floor process there exists a wealth which is optimal for the unconstrained problem and dominates any given fraction of this floor process. In consequence, the floor constraint does not affect the optimal value of the long-term optimisation problem. This underlines the simplifying nature of growth rate maximisation criterion \eqref{eq:cer_def}. It clearly does not distinguish between two wealth processes which agree from some point in time onwards. As it turns out, it also washes away more subtle features of finite horizon problems. Our results are obtained in a general semimartingale setup and extend\footnote{This research was conducted as part of the first author's Ph.D.\ thesis and was in fact done independently of \cite{sek}.} Sekine \cite{sek} who dealt with the power utility function and considered diffusion driven market. We note also that our results do not contradict Davis and Lleo \cite{DavisLleo:08} who applied risk-sensitive control tools but with $F_T=\log (V_T/G_T)$ and hence had a different objective.

The second question we consider in this paper relates to the asymptotic optimality of solutions to finite horizon utility maximisation problems as defined in \eqref{eq:cer_def}.
The long-term behaviour of an investor whose preferences are characterised via a utility function was studied in many influential papers. In particular, the so called \emph{turnpike theorems} establish
necessary and sufficient conditions for the optimal strategy for finite horizon to converge as the horizon tends to infinity, see \cite{Mos,Lel,Hak} for discrete market setting  and \cite{CoxHu,HuZa,DyRoBa,GuRo} for continuous market setting.
More recently, Guasoni and Robertson \cite{GuasoniRobertson:11} described the concept of long-run optimality of the wealth process. The wealth process $V$ is called \emph{long-run optimal} if the rate of growth of value functions for finite horizon problems converges to the CER, i.e.\ the optimal asymptotic growth rate of expected utility of $V$ in \eqref{eq:cer_def}.
We consider this notion in the framework of constrained optimisation. Namely, we provide conditions for the wealth which solves constrained long-term optimisation problem to be long-run optimal.

The paper is organised as follows. Below, we introduce a general market setup used throughout. In Section \ref{sec:longterm} we solve the long-term investor problem under floor constraint. And in Section \ref{sec:longrun} the long-run property is studied in the constrained framework.

\subsection{Market setup}\label{sec:market}

We consider a frictionless market defined on the filtered probability space $(\Omega, \mathcal{F},$ $ (\mathcal{F}_t), \mathbb{P})$ satisfying the usual hypothesis. All considered processes are implicitly taken right-continuous with left-limits (c\`adl\`ag).
The underlying assets are represented as a vector of strictly positive semimartingales $\tilde{S}_t = (\tilde{S}^0_t,.., \tilde{S}^n_t)$. We fix $\tilde{S}^0_t=N_t$ to be the baseline asset or our numeraire. We stress that here $N$ is arbitrary and could be the riskless asset, a stock, but also a portfolio process. The underlying assets in units of $N$ are expressed as $S:= (1, S^1,..,S^d)$ with $S^i_t := \tilde{S}^i_t/N_t$. All assets and portoflios will be expressed in units of $N$.

The investor is assumed to trade in the usual self-financing way. The set of all admissible investment strategies is given as :
\begin{defn} \label{def:wealth1}
An adapted semimartingale $(V_t)$ is called a \emph{wealth process} if it is strictly positive and there exists an $(\F_t)$--predictable process $\pi=(\pi^1_t,\ldots,\pi^d_t)$ such that $V_t=V_0+\sum_{i=1}^d \int_0^t \pi^i_u \td S^i_u$, where the (vector) integral is assumed to be well-defined. The set of wealth processes with $V_0=v_0$ is denoted $\A(v_0)$.
\end{defn}

\subsection{Az\'ema--Yor processes} \label{sec:azema-yor}
We recall briefly Az\'ema--Yor processes which will be of use later in the paper. They were studied recently by Carraro, El Karoui and Ob\l\'oj \cite{CEO} and we refer the reader to \cite{CEO} for further details.

\begin{prop}[Carraro, El Karoui and Ob\l\'oj \cite{CEO}]\label{prop:AY}
Let $F'$ be a locally bounded function, $F(x)=F(x_0)+\int_{x_0}^x F'(u)\td u$, and $(X_t)$ a max-continuous $(\F_t)$--semimartingale. The associated Az\'ema--Yor process is given via
\begin{equation}\label{eq:AYdef}
M^F_t(X):=F(\ovl X_t)-F'(\ovl X_t)(\ovl X_t - X_t) = F(X_0)+\int_0^t F'(\ovl X_u)\td X_u,
\end{equation}
where $\ovl X_t := \sup_{u\leq t} X_u$. Further
\begin{itemize}
\item[(i)] if $F'\geq 0$ then $\ovl{M^F_t(X)}=F(\ovl X_t)$,
\item[(ii)] if $F'>0$ then $M^K_t(M^F_t(X))=X_t$ with $K=F^{-1}$ the inverse of $F$,
\item[(iii)] if $F$ is concave then $M^F_t(X)\geq F(X_t)$,
\item[(iv)] if $(V_t)_{t\geq 0} \in \A(v_0)$ and $F(v_0)= v_0 > 0$, $F'\geq 0$, then $X_t := M^F_t(V) \in \A(v_0)$ and $X_t> w(\ovl X_t)$ where $w(x)=x-K(x)/K'(x)$, $K=F^{-1}$.
\end{itemize}
\end{prop}
The results follows directly from Section 2 in \cite{CEO} and only the last property requires an extra argument. Using \eqref{eq:AYdef} one obtains that $X_t\geq 0$ and it satisfies Definition \ref{def:wealth1} with
\begin{equation*}
\pi^X_t=F'(\ovl{V}_t)\pi^V_t,
\end{equation*}
where $\pi^X_t$ and $\pi^V_t$ are vectors.

\section{Optimal long-term investment subject to floor constraint} \label{sec:longterm}

We consider now the infinite horizon problem of an investor who aims to maximise her asymptotic growth rate of expected utility of terminal wealth \eqref{eq:cer_def} but is subject to a floor constraint. Our main theorem provides the connection between constrained and unconstrained problems. We are able to show that the floor constraint does not decrease the value function for a wide class of floor processes considered. More precisely, for any floor process which admits at least one wealth process dominating it, there exists an optimal solution to the unconstrained problem which dominates any given fraction of this floor process.

For a given adapted non-negative semimartingale process $(G_t)_{t\geq 0 }$ we consider a subclass of all wealth processes defined as follows
\begin{equation}
\mathcal{A}_G(v_0) := \{ \left(V_t\right) \in \A (v_0): \quad  V_t \geq G_t, \quad t \geq 0 \}. \nonumber
\end{equation}

We consider function $U$ which satisfies the following
\begin{asmp}\label{ass:U1}
Function $U$ is non-decreasing, concave and it is either strictly positive, or it is strictly negative. It satisfies
\begin{equation}
\limsup_{x \rightarrow \infty} \frac{xU'_+(x)}{\left| U(x) \right|} < \infty. \label{eq:KS_cond_on_U}
\end{equation}
\end{asmp}


We now state our main result, which generalises the observations in Sekine \cite{sek}. Below when using $\R_{U}$ from \eqref{eq:cer_def}, and throughout, we extend $\log$ to $\re\setminus \{0\}$ via $\log(x)=-\log(-x)$.

\begin{thrm} \label{thrm:main}
 Let $U$ satisfy Assumption \ref{ass:U1} and $v_0 >0 $ be an initial wealth.
Consider a floor process $(G_t)_{t \geq 0}$ such that $G_t \leq X_t \in \mathcal{A}(v_0(1-\varepsilon))$ for a wealth process $X$ and some $0 < \varepsilon < 1$.
Then
 \begin{equation}
\sup_{V \in \mathcal{A}(v_0)}\R_U(V) = \sup_{V \in \A_G(v_0)}\R_U(V)\ . \label{eq:main_thr}
 \end{equation}
Further, if the left hand side in \eqref{eq:main_thr} is finite and achieved by some wealth process $\hat{\xi}$ then we can assume that $\hat{\xi}_t \geq c > 0$  a.s.\ for some $c \in \mathbb{R}^+$ and all $t\geq 0$. The right hand side in \eqref{eq:main_thr} is then maximised by $\hat{V}_t := \varepsilon \hat{\xi}_t + X_t$.
\end{thrm}

\begin{proof}
Suppose that the LHS in \eqref{eq:main_thr} is finite and achieved by an optimal wealth process $(\hat{\eta}_t)_{t\geq 0} \in \A(v_0)$.  For any $\delta \in (0,1)$ we consider $\hat{\xi}_t := \delta v_0 + (1-\delta )\hat{\eta}_t \geq \delta v_0$ which belongs to $\mathcal{A}(v_0)$.\footnote{Note that in fact $\hat{\xi}=M^F(\eta)$ with $F(x)=\delta v_0+(1-\delta)x$ which corresponds to $w(x)\equiv \delta v_0$ in $(iv)$ in Proposition \ref{prop:AY}.}
Using Lemma A.3 from \cite{ChOb} for function $U$ which satisfies Assumption \ref{ass:U1}, we obtain that for any $x_0 > 0$ there exists $\gamma \in \mathbb{R}$ such that
\begin{equation*}
U(x) \leq U(\lambda x) \leq \lambda ^ {\gamma} U(x) \quad \text{for any}\quad  \lambda > 1, x \geq x_0.
\end{equation*}
Applying the above with $x_0=\delta v_0$ we obtain
\begin{equation*}
U(\hat{\xi}_t) \geq (1-\delta)^{\gamma'} U\left(\frac{\delta}{1-\delta} v_0 + \hat{\eta}_t\right) \geq (1-\delta)^{\gamma'} U(\hat{\eta}_t).
\end{equation*}
Thus, $\R_U(\hat\xi) \geq \R_U(\hat{\eta})$, which concludes the optimality of $\hat{\xi}$.

Consider now the process $\hat{V}_t := \varepsilon \hat{\xi}_t + X_t$. It is a wealth process from the class $\mathcal{A}(v_0)$\footnote{This turns out to be the same strategy as used by Sekine in \cite{sek}.}.
From nonnegativity of $\hat{\xi}$ we deduce
$$\hat{V}_t \geq X_t \geq G_t.$$
Therefore, $(\hat{V}_t)_{t \geq 0} \in \A_G(v_0)$.

Now, we need to show the optimality of the wealth process $\hat{V}$. Note that if we have two wealth processes such that $\theta^1_t \geq \theta^2_t$ for all $t \geq 0$ then $\R(\theta^1) \geq \R(\theta^2)$. This is due to the fact that $U$ is non-decreasing.
And as $\hat{V}_t \geq \varepsilon \hat{\xi}_t$ we deduce that $\R(\hat{V}) \geq \R(\varepsilon \hat{\xi})$ .

As $\hat{\xi}_t \geq c$ for some $c >0$, we deduce that $U(\varepsilon \hat{\xi}_t) \geq \varepsilon ^ {\gamma} U(\hat{\xi}_t)$ again by Lemma A.3 from \cite{ChOb} with $x_0 = c \varepsilon$. Thus, $\R(\varepsilon \hat{\xi}) = \R(\hat{\xi})$ and we conclude the proof of optimality of $\hat{V}$.

The same arguments, possibly applied to a sequence of wealth processes, show that \eqref{eq:main_thr} holds in all generality.
\end{proof}

\begin{remark} \label{remark:longrun}
From the proof of Theorem \ref{thrm:main} we obtain that
 $\e U(\hat{V}_T)<\infty$ and $\R_U(\hat{V})<\infty$ if and only if the same holds for $\hat{\xi}$ and then
\begin{equation*}
\limsup_{T \rightarrow \infty} \frac{1}{T} \left( \log \e U(\hat{\xi}_T) - \log \e U(\hat{V}_T)\right) = 0 \ .
\end{equation*}
\end{remark}

\begin{proof}
The first assertion is clear by \eqref{eq:main_thr} and since $U(\hat{V}_T)\geq U(\varepsilon \hat{\xi}_T)\geq \varepsilon^\gamma U(\hat{\xi}_T)$. Further
\begin{equation*}
\limsup_{T \rightarrow \infty} \frac{1}{T} \left( \log \e U(\hat{\xi}_T) - \log \e U(\hat{V}_T)\right) \geq \limsup_{T \rightarrow \infty} \frac{1}{T} \log \e U(\hat{\xi}_T) - \limsup_{T \rightarrow \infty} \frac{1}{T} \log \e U(\hat{V}_T) = 0,
\end{equation*}
and
\begin{equation*}
\begin{split}
\limsup_{T \rightarrow \infty} \frac{1}{T} \left( \log \e U(\hat{\xi}_T) - \log \e U(\hat{V}_T)\right) &\leq \limsup_{T \rightarrow \infty} \frac{1}{T} \left( \log \e U(\hat{\xi}_T) - \log \varepsilon ^{\gamma}\e U(\hat{\xi}_T)\right) \\
&\leq \limsup_{T \rightarrow \infty} \frac{1}{T} (-\gamma \log \varepsilon)=0.
\end{split}
\end{equation*}
\end{proof}

\begin{remark}
 Combining this result and Theorem 4.2 from \cite{ChOb} we are able to connect in an explicit manner three problems: the unconstrained problem, the problem with a drawdown constraint and the problem with a floor constraint.

More precisely, for a continuous stock process $(S_t)_{t\geq 0}$ and under assumptions of Theorem 4.2 from \cite{ChOb} the solution to the drawdown constrained problem is an explicit Az{\`e}ma-Yor transformation of the solution to floor constrained problem, since the latter also solves the unconstrained problem.
\end{remark}

\section{Long-run properties} \label{sec:longrun}

We turn now to studying the behaviour of the value function of the finite horizon problem when the horizon becomes distant. As discussed in the Introduction, such topics have been studied primarily in the context of \emph{turnpike theorem}. Here we look at \emph{long-run optimality} in the sense of Guasoni and Robertson \cite{GuasoniRobertson:11}.  We establish long-run optimality results in the constrained framework, when there is either a floor constraint, or a drawdown constraint imposed on the wealth processes.

\subsection{Definition of long-run optimality}
An investment strategy is called long-run optimal if the rate of growth of the value function for finite horizon problem converges to the certainty equivalent rate (CER) in \eqref{eq:cer_def}. Specifically, let us first define classes of admissible wealth processes up to time horizon $T$.
For a given floor process $(G_t)_{0\leq t \leq T}$ we define a subclass of $\A(v_0)$ as
$$\mathcal{A}_G(T)(v_0) := \{V \in \A(v_0) \quad \text{ s.t. } \quad V_t \geq G_t, \quad 0\leq t \leq T\}.$$
For a given function $w$ we define
\begin{equation*}
\A^w(T)(v_0) := \{ V \in \A(v_0) \quad \text{ s.t. } \quad \min \{ V_{t-}, V_t\} > w(\overline{V}_t) \quad 0\leq t\leq T\},
\end{equation*}
where $\overline{V}_t := \sup_{s \leq t} V_s$.
We simply write $\A^w(v_0)$ and $\mathcal{A}_G(v_0)$ when $T = \infty$.

\begin{defn} \label{defn:long-run}
Consider $U$ satisfying Assumption \ref{ass:U1}. A wealth process $\hat{V} \in \mathcal{B}$ is called $\mathcal{B}$-\textit{long-run optimal} if
$\sup_{V \in \mathcal{B}_T} \e U(V_T)$ is finite and 
\begin{equation}
\limsup_{T \rightarrow \infty} \left( \frac{1}{T} \log \sup_{V \in \mathcal{B}_T} \e U(V_T) - \frac{1}{T} \log \e U(\hat{V}_T)\right) = 0, \label{eq:long_run_opt}
\end{equation}
where $\mathcal{B}_T = \A(v_0)$, $\A_G(T)(v_0)$ or $\A^w(T)(v_0)$, and $\mathcal{B} = \A(v_0)$, $\A_G(v_0)$ or $\A^w(v_0)$, respectively.
\end{defn}

For a financial interpretation of this definition and intuition behind we refer to \cite{GuasoniRobertson:11} where the authors define the \textit{certainty equivalent loss} $l: \left[ 0, \infty \right) \rightarrow \mathbb{R}$ of a wealth process $(V_t)_{t\geq 0}$ by
\begin{equation*}
 \e U(e^{l_T T} V_T) = \sup_{Y \in \A(v_0)} \e U(Y_T).
\end{equation*}
Similar ideas were used by Grossman and Villa \cite{GrossmanVila:92} to measure the so-called cost of myopia.
One can see that for a process $V \in \A(v_0)$ such that $V \geq c > 0$: if $l_T(V) \rightarrow 0$ as $T \rightarrow \infty$ then $V$ is $\A(v_0)$-long-run optimal. Indeed, as $V \geq c > 0$ we are able to use Lemma A.3 from \cite{ChOb} to derive for some $\gamma \in \mathbb{R}$ that $U(e^{l_T T}V_T) \leq e^{\gamma l_T T} U(V_T)$. Thus,
\begin{equation*}
0\leq \limsup_{T \rightarrow \infty} \left( \frac{1}{T} \log \sup_{Y \in \A(v_0)} \e U(Y_T) - \frac{1}{T} \log \e U(V_T)\right) \leq \limsup_{T\rightarrow \infty} \gamma l_T = 0
\end{equation*}

In other words, vanishing certainty equivalent loss is a sufficient condition for long-run optimality.

\subsection{Floor constraint}
In this subsection we construct the long-run optimal portfolio satisfying a floor constraint using the long-run optimal portfolio for an unconstrained problem. The main result for this section is the following:
\begin{prop}
Let $(\hat{X}_t)_{t\geq 0}$ be $\A(v_0)$-long-run optimal for some utility function $U$ which satisfies Assumption \ref{ass:U1}. Then, for any floor process $(G_t)_{t\geq0}$ such that  $G_t \leq X_t$ for some  $0< \varepsilon < 1$ and some $X \in \A(v_0(1-\varepsilon))$, the process $\hat{V}_t := \varepsilon \hat{X}_t + X_t$ is $\A_G(v_0)$-long-run optimal.
\end{prop}

\begin{proof}
By Remark \ref{remark:longrun} one obtains
\begin{eqnarray*}
0 &\leq &\limsup_{T \rightarrow \infty} \left( \frac{1}{T} \log \sup_{V \in \A_G(T)(v_0)} \e U(V_T) - \frac{1}{T} \log \e U(\hat{V}_T)\right) \\ 
&\leq &\limsup_{T \rightarrow \infty} \left( \frac{1}{T} \log \sup_{V \in \A_G(T)(v_0)} \e U(V_T) - \frac{1}{T} \log \e U(\hat{X}_T)\right) \\
&\leq & \limsup_{T \rightarrow \infty} \left( \frac{1}{T} \log \sup_{V \in \A(v_0)} \e U(V_T) - \frac{1}{T} \log \e U(\hat{X}_T)\right).
\end{eqnarray*}
where we used $\A_G(T)(v_0) \subset \A(v_0)$. The right hand side is equal to zero as $\hat{X}$ is $\A(v_0)$-long-run optimal. We conclude that $\hat{V}$ is $\A_G(v_0)$-long-run optimal.
\end{proof}

\subsection{Drawdown constraint}
We now turn to the long-run optimality of the solution to the drawdown constrained problem. We assume that all $S^i$ are continuous\footnote{We restrict our attention to continuous assets to avoid notational technicalities but the results naturally hold in the setup of \emph{max-continuous} assets and wealth processes as in \cite{ChOb}.} and hence all processes in $\A(v_0)$ are also continuous. We first recall the necessary definitions following closely \cite{ChOb}.
\begin{defn}\label{def:DDfunc}
We say that $w$ is a \emph{drawdown function} if it is non-decreasing and
\begin{equation}\label{eq:w-assumption}
\exists \alpha: 0 \leq w(x)/x \leq \alpha < 1,\quad x\geq 0.
\end{equation}
\end{defn}
Define
\begin{equation}\label{eq:Kdef}
K_w(x) := v_0 \exp \left( \int_{v_0}^x \frac{1}{u - w(u)} \td u \right), x \geq v_0,
\end{equation}
which is continuous and strictly increasing and has a well defined inverse $F_w:=K_w^{-1}:[v_0, \infty)\to [v_0, \infty)$. It follows from
(ii) and (iv) in Proposition \ref{prop:AY}, or more generally from Proposition 3.2 in \cite{ChOb}, that if $X\in \A(v_0)$ then $V=M^F(X)\in \A^w(v_0)$
and $X=M^K(V)$ establishing a bijection between $\A(v_0)$ and $\A^w(v_0)$ (or $\A^w(T)(v_0)$). We shall exploit this relation on several occasions in the sequel.

\begin{asmp}\label{ass:utility}
Assume that, for some $\varepsilon>0$, $U$ satisfies either
\begin{equation*}
\frac{U(x)}{x^{\varepsilon}} \vax \infty,\quad \textrm{and $U$ is strictly positive on }(0,\infty),
\end{equation*}
or
\begin{equation*}
 U(x)x^{\varepsilon} \vax 0, \quad \textrm{and $U$ is strictly negative on }(0,\infty).
\end{equation*}
\end{asmp}

\begin{thrm} \label{thrm:long-run}
Let $w$ be a drawdown function and $U$ satisfy Assumptions \ref{ass:U1} and \ref{ass:utility}.
Assume that $\limsup_{T \rightarrow \infty} \frac{1}{T} \log \sup_{V \in \A(v_0)}\e \left(U \circ F_w\left(V_T\right)\right)^{1+\delta} < \infty$ for some $\delta > 0$. Then
\begin{itemize}
\item[(i)] If $(\hat{\xi}_t)_{t \geq 0}$ is $\A(v_0)$-long-run optimal with utility function $U \circ F_w$ then
$M^{F_w}_t(\hat{\xi})$ is $\A^w(v_0)$-long-run optimal with utility function $U$.
\item[(ii)] Suppose $(\hat{\xi}_t)_{t \geq 0}$ solves \eqref{eq:cer_def} among all $V\in \A(v_0)$ and
\begin{equation*}
\limsup_{T \rightarrow \infty} \frac{1}{T} \log \e U \circ F_w(\hat{\xi}_T) = \liminf_{T \rightarrow \infty} \frac{1}{T} \log \e U \circ F_w(\hat{\xi}_T).
\end{equation*}
Recall that, by Theorem 4.1 from \cite{ChOb}, $\hat{X}:=M^{F_w}(\hat{\xi}) \in \A^w(v_0)$ solves \eqref{eq:cer_def} among all $X\in \A^w(v_0)$. \\
If $\hat{X}$ is $\A^w(v_0)$-long-run optimal with utility function $U$ then $\hat{\xi}$ is $\A(v_0)$-long-run optimal with utility function $U \circ F_w$.
\end{itemize}
\end{thrm}
\begin{proof}
Recall that by Lemma A.3 in \cite{ChOb} for any $\tilde U$ which satisfies Assumption \ref{ass:U1} and any $x_0>0$ there exists 
$\gamma \in \mathbb{R}$ such that for all $y\geq x_0$ and all $c>1$:
\begin{equation}\label{eq:A3}
\tilde U(y) \leq \tilde U(cy) \leq c^{\gamma} \tilde U(y).
 \end{equation}
We will use this several times below, in particular for $\tilde U=U$ or $U\circ F_w$. This fact also implies that (see Lemma A.1 in \cite{ChOb}) we may (and will) assume that $\hat{\xi}_t\geq v_0/2$ for all $t\geq 0$.

We recall the properties of Az\'ema--Yor processes given in Proposition \ref{prop:AY} and their use to obtain a bijection between $\A^w(T)(v_0)$ and $\A(v_0)$. Also, introduce $w_{\varepsilon}(x) = \varepsilon x$, the associated $K_{w_{\varepsilon}}(x) = v^{\varepsilon/(1-\varepsilon)}_0 x^{1/(1-\varepsilon)}$ and its inverse $F_{w_\varepsilon}$. It follows that 
\begin{equation*}
\begin{split}
 \sup_{V \in \A^w(T)(v_0)} \e U(V_T) &\leq  \sup_{Y \in \A(v_0)} \e U \circ F_w(\ovl{Y}_T) = \sup_{Y\in \A^{w_\varepsilon}(v_0)}\e U\circ F_w\left(K_{w_\epsilon}(\ovl{Y}_T)\right)\\
& \leq \sup_{Y\in \A^{w_\varepsilon}(v_0):Y\geq \varepsilon v_0}\e U\circ F_w\left(K_{w_\epsilon}\left(\frac{1}{\varepsilon}Y_T\right)\right)\\
&\leq \sup_{Y\in \A(v_0):Y\geq \varepsilon v_0}\e U\circ F_w\left(c_\varepsilon Y_T^{\frac{1}{1-\varepsilon}}\right),\textrm{ where }c_\varepsilon = \left(\frac{v_0^\varepsilon}{\varepsilon}\right)^{\frac{1}{1-\varepsilon}}\\
&\leq c_\varepsilon^\gamma \sup_{Y\in \A(v_0):Y\geq \varepsilon v_0}\e U\circ F_w\left(Y_T^{\frac{1}{1-\varepsilon}}\right)
\end{split}
\end{equation*}
where we took $\varepsilon<\frac{1}{2}$ small enough so that $c_\varepsilon>1$ and applied \eqref{eq:A3} to $\tilde U=U\circ F_w$ with $x_0=\varepsilon v_0$.

On the other hand, using the property of Az\'ema--Yor processes for the concave function $F_w$, we get
\begin{equation*}
 \e U(M^{F_w}_T(\hat{\xi})) \geq  \e U \circ F_w(\hat{\xi}_T).
\end{equation*}
Combining the two inequalities we deduce that
\begin{eqnarray}
0 & \leq & \limsup_{T\rightarrow \infty} \frac{1}{T}\left[ \log \sup_{V \in \A^w(T)(v_0)} \e U(V_T) - \log  \e U(M^{F_w}_T(\hat{\xi}))\right] \label{eq:main_ineq} \\
& \leq & \limsup_{T\rightarrow \infty} \frac{1}{T}\left[ \log \sup_{Y \in \A(v_0):Y\geq \varepsilon v_0}\e U \circ F_w (Y_T) - \log  \e U \circ F_w(\hat{\xi}_T)\right] + \nonumber \\
&& +  \limsup_{T\rightarrow \infty} \frac{1}{T}\left[ \log c_\varepsilon^\gamma \sup_{Y \in \A(v_0):Y\geq \varepsilon v_0}\e U \circ F_w((Y_T)^{\frac{1}{1-\varepsilon}}) - \log \sup_{Y \in \A(v_0):Y\geq \varepsilon v_0}\e U \circ F_w (Y_T)\right],\nonumber
\end{eqnarray}
where the first $\limsup$ is zero by long-run optimality of $\hat{\xi}$ and the second one is non-negative (this can be seen using the restriction $Y\geq \varepsilon v_0$ and applying \eqref{eq:A3}). We will now argue that it is bounded by a constant times $\varepsilon/(1-\varepsilon)$ and hence goes to zero as $\varepsilon$ does.

To this end, we will show that there exists $K > 0$ such that for all $\tilde{U}$ such that $\tilde{U}(x) \leq \kappa U \circ F_w (x)^{1+\delta}$ for a constant $\kappa$ and all $x\geq 1$, where $\delta>0$ is given in the statement, we have
\begin{equation}\label{eq:def_tilde_c}
\tilde{c} :=\limsup_{T \rightarrow \infty} \frac{1}{T} \left[ \log \sup_{V \in \A(v_0):V\geq \varepsilon v_0} \e \tilde{U}(V_T) - \log \sup_{V \in \A(v_0)V\geq \varepsilon v_0} \e \tilde{U}(V_T) \mathbf{1}_{V_T \leq K^T}\right]=0
\end{equation}
is equal to zero. Indeed, let $(V^T_t)_{t\geq 0}$ achieve the supremum in $\sup_{V \in \A(v_0)} \e \tilde{U}(V_T)$ (if such process does not exist then choose such $(V^T_t)_{t\geq0}$ that $\e U(V^T_T)$ differs from the supremum by less than $1/T$), then
$$0 \leq \tilde{c} \leq \limsup_{T \rightarrow \infty} \frac{1}{T} \left[ \log \e \tilde{U}(V^T_T) - \log \e \tilde{U}(V^T_T) \mathbf{1}_{V^T_T \leq K^T}\right].$$
 The argument then follows the lines of the proof of Lemma A.2 from \cite{ChOb}, where we set $\xi_T := V^T_T$ and we set $C_G = \limsup_{T \rightarrow \infty} \frac{1}{T} \log \sup_{V}\e \left(U \circ F_w(V_T)\right)^{1+\delta}$ where $V\in \A(v_0)$ with $V\geq \varepsilon v_0$. 
 
Thus, using \eqref{eq:def_tilde_c} for $\tilde{U} = U \circ F_w(x^{\frac{1}{1-\varepsilon}})$ (by \eqref{eq:A3} and taking $\varepsilon$ small enough) and for $\tilde{U} = U \circ F_w(x)$ we continue \eqref{eq:main_ineq} to obtain
\begin{eqnarray}
0&\leq & \limsup_{T\rightarrow \infty} \frac{1}{T}\left[ \log \sup_{V \in \A^w(T)(v_0)} \e U(V_T) - \log  \e U(M^{F_w}_T(\hat{\xi}))\right] \nonumber \\
& \leq & \limsup_{T\rightarrow \infty} \frac{1}{T}\left [ \log c_\varepsilon^\gamma\sup_{Y \in \A(v_0):Y\geq \varepsilon v_0}\e U \circ F_w((Y_T)^{\frac{1}{1-\varepsilon}}) \mathbf{1}_{Y_T \leq K^T}\right. \nonumber \\
&& \left.- \log \sup_{Y \in \A(v_0): Y\geq \varepsilon v_0}\e U \circ F_w (Y_T)\mathbf{1}_{Y_T \leq K^T}\right] \nonumber \\
& \leq & \frac{\varepsilon}{1-\varepsilon} |\gamma| \log K \label{eq:long-run_dd}
\end{eqnarray}
where we used that $U \circ F_w (y^{\frac{1}{1-\varepsilon}}) \leq U\circ F_w(y (y/\varepsilon v_0)^{\varepsilon/(1-\varepsilon)})\leq (y/\varepsilon v_0)^{\gamma \varepsilon/(1-\varepsilon)}U \circ F_w (y)$ for $y \geq \varepsilon v_0$ by \eqref{eq:A3}. Finally, letting $\varepsilon \to 0$ we obtain \begin{equation*}
\limsup_{T\rightarrow \infty} \frac{1}{T}\left[ \log \sup_{V \in \A^w(T)(v_0)} \e U(V_T) - \log  \e U(M^{F_w}_T(\hat{\xi}))\right]  = 0
\end{equation*}
which establishes $(i)$. For $(ii)$ it suffices to write
\begin{eqnarray*}
0&\leq & \limsup_{T\rightarrow \infty} \frac{1}{T}\left[ \log \sup_{Y \in \A(v_0)} \e U \circ F_w(Y_T) - \log  \e U \circ F_w(\hat{\xi}_T)\right] \\
& \leq & \limsup_{T\rightarrow \infty} \frac{1}{T}\left[ \log \sup_{\zeta \in \A^w(T)(v_0)} \e U (\zeta_T) - \log  \e U (M^{F_w}_T(\hat{\xi}))\right]\\
&& + \limsup_{T\rightarrow \infty} \frac{1}{T}\left[ \log \e U (M^{F_w}_T(\hat{\xi})) - \log \e U \circ F_w (\hat{\xi}_T)\right] \\
& = & 0+ \cer^w_U(v_0) - \cer_{U\circ F_w}(v_0) = 0,
\end{eqnarray*}
where in last equation we used the fact that $\limsup (A - B) \leq \limsup A - \liminf B = \limsup A - \limsup B$ when $\limsup B = \liminf B$.

\end{proof}
\subsection{Example: complete market model with deterministic coefficients}
We consider now the classical complete financial market model with deterministic coefficients. $W_t=(W^1_t,\ldots,W^d_t)\tr$ is a standard $d$-dimensional Brownian motion and $(\F_t)_{t \geq 0}$ is the right-continuous augmentation of its natural filtration. Here $\tr$ denotes vector transpose. Each asset follows dynamics given by
$$\frac{\td S^i_t}{S^i_t}= \mu^i_t\td t + \sum_{j=1}^d \sigma^{ij}_t\td W^j_t,\quad S^i_0=s^i_0>0$$
where $\mu^i_t$ and $\sigma^{ij}_t$ are bounded deterministic functions and $\sigma_t$ is invertible. Recall Definition \ref{def:wealth1} of wealth process and let $\tilde \pi_t^i:= \pi^i_t S^i_t/V_t$ be the proportion of wealth invested in the $i^{\textrm{th}}$ asset so that $\td V_t= \sum_{i=1}^d \tilde\pi^i_t V_t \frac{\td (S^i_t)}{S_t^i}$.
The market price of risk is given as $\theta_t:= \sigma^{-1}\mu_t$. We assume $\theta_t$ is also bounded and that
$$||\theta^*||^2:= \lim_{T\to\infty}\frac{1}{T}\log \int_0^T ||\theta_u||^2\td u\quad \textrm{ is well defined and finite.}$$

We consider the problem of maximising the expected utility of discounted wealth at a given horizon $T$. The solution is obtained using, by now standard, convex duality arguments, see Karatzas, Lehoczky and Shreve \cite{KaLeSh} or Karatzas and Shreve \cite[pp.~97--118]{KaSh}. It involves the state price density
$$Z_t:= \exp\left\{-\int_0^t \theta_u\tr \td W_u - \frac{1}{2}\int_0^t ||\theta_u||^2 \td u\right\},\quad t\geq 0,$$
which is a $\p$--martingale and defines the unique risk neutral measure $\Q$ on $\F_T$ via $\frac{\td \Q}{\td \p}|_{\F_T}=Z_T$.
The value function for the utility function $U_p(x) = \frac{1}{p}x^p$ equals
\begin{eqnarray*}
 V(v_0,T,p) &=& \sup_{V \in \A(v_0)} \e U_p(V_T) = U_p(v_0) \left( \e Z^{-\frac{p}{1-p}}_T\right)^{1-p} \\
&=& U_p(v_0) \exp \left\{\frac{p}{2(1-p)} \int_0^T ||\theta_u||^2 du\right\}.
\end{eqnarray*}
Moreover, the optimal wealth process $V^*$ which is characterised via $\tilde{\pi}^*_t = \frac{1}{1-p} \theta'_t \sigma^{-1}_t$ is independent of the horizon $T$ and, therefore, is long-run optimal for the unconstrained problem.

Considering the linear drawdown constraint, $w(x)  = \alpha x$, we check that the assumptions of the first part in Theorem \ref{thrm:long-run} hold  and hence we deduce long-run optimality of the wealth process $X_t=M^{F_w}_t(V^*)$ for $w$-drawdown constrained problems, since $V^*$ solves the unconstrained problem. Moreover, the following asymptotics of the finite horizon problem with drawdown constraint is obtained using Theorem \ref{thrm:long-run}
$$\log \sup_{V \in \A^w(T)(v_0)} \e U_p(V_T) = T\left(\frac{|p|(1-\alpha)}{2(1-p(1-\alpha))}||\theta^*||^2 + o(1)\right)$$

This asymptotic result can be sharpened as follows:
\begin{prop}
$$\log \sup_{V \in \A^w(T)(v_0)} \e U_p(V_T) = \frac{|p|(1-\alpha)}{2(1-p(1-\alpha))}\int_0^T||\theta_t||^2dt + O\left(\log T\right).$$
\end{prop}
\begin{proof}
With no loss of generality we put $v_0 = 1$. By property of the Az\'ema--Yor processes for a concave function $F_w$, we obtain
\begin{equation*}
\sup_{V \in \A^w(T)(1)} \e U_p(V_T) \geq \sup_{X \in \A(1)} \e U_p \circ F_w(X_T) = (1-\alpha) V(1,T,p(1-\alpha)).
\end{equation*}
To obtain a reverse inequality we need to perturb the drawdown constraint. We write $F_\alpha=F_w$ and for a small $\varepsilon>0$ consider $K_{\varepsilon}(x)= x^{\frac{1}{1-\varepsilon}}$ and its inverse $F_{\varepsilon}(y) = y^{1-\varepsilon}$ which correspond to the drawdown function $w_{\varepsilon}(x) = \varepsilon x$. Naturally, since $\A^{w_{\varepsilon}}(T)(1) \subseteq \A(1)$, we have
\begin{equation*}
\sup_{X \in \A(1)} \e U_p \circ F_\alpha \circ K_{\varepsilon}(X_T) \geq \sup_{X \in \A^{w_{\varepsilon}}(T)(1)} \e U_p \circ F_\alpha \circ K_{\varepsilon}(X_T), \label{eq:dd_finite}
\end{equation*}
and the right hand side can be rewritten as
\begin{equation*}
\sup_{X \in \A^{w_{\varepsilon}}(T)(1)} \e U_p \circ F_\alpha \circ K_{\varepsilon}(X_T) = \sup_{Y \in \A(1)} \e U_p \circ F_\alpha \circ K_{\varepsilon}(M^{F_{\varepsilon}}_T(Y)).
\end{equation*}
Using the drawdown constraint property of $M^{F_{\varepsilon}}_T(Y)$ we deduce
\begin{eqnarray*}
\sup_{X \in \A(1)} \e U_p \circ F_\alpha \circ K_{\varepsilon}(X_T) &\geq & \varepsilon^{\frac{p(1-\alpha)}{1-\varepsilon}} \sup_{Y \in \A(1)} \e U_p \circ F_\alpha (\overline{Y}_T)  \\
& \geq & \varepsilon^{\frac{p(1-\alpha)}{1-\varepsilon}} \sup_{V \in \A^w(T)(1)} \e U_p(V_T) .
\end{eqnarray*}

Thus, we obtain inequality
\begin{equation*}
(1-\alpha) V\left(1,T,p(1-\alpha)\right) \leq \sup_{V \in \A^w(T)(1)} \e U_p(V_T) \leq \frac{1-\alpha}{1-\varepsilon} \varepsilon^{-\frac{p(1-\alpha)}{1-\varepsilon}}V\left(1,T, \frac{p(1-\alpha)}{1-\varepsilon}\right)
\end{equation*}

Taking logarithm we obtain
\begin{eqnarray*}
&& \log \frac{1}{p} + \frac{|p|(1-\alpha)}{2(1-p(1-\alpha))}\int_0^T||\theta_t||^2dt \leq \log \sup_{V \in \A^w(T)(1)} \e U_p(V_T) \\
&& \leq \log \frac{1}{p} + \frac{|p|(1-\alpha)}{2(1-p(1-\alpha) - \varepsilon)}\int_0^T||\theta_t||^2dt - \frac{p(1-\alpha)}{1-\varepsilon}\log \varepsilon.
\end{eqnarray*}
Now, taking $\varepsilon = \frac{1}{T}$ we obtain the required asymptotics.

\end{proof}

\bibliographystyle{acm}
\bibliography{bibliography}
\end{document}